%% file: root.tex
\documentclass[letterpaper, 10 pt, conference]{ieeeconf}  

\IEEEoverridecommandlockouts                              
\usepackage{geometry}
\usepackage{geometry}
\usepackage{epsfig} 
\usepackage{times} 
\usepackage{amsmath} 
\usepackage{amssymb}  

\usepackage{color}
\usepackage[algo2e]{algorithm2e}

\usepackage{subcaption}
\usepackage{float}
\usepackage{booktabs} 

\usepackage{graphicx}
\usepackage{afterpage}
\usepackage[acronym,smallcaps]{glossaries}
\usepackage{enumerate}
\usepackage{hyperref}
\usepackage[english]{babel}
\newtheorem{theorem}{\textbf{Theorem}}

\newtheorem{lemma}{\textbf{Lemma}}
\newtheorem{definition}{\textbf{Definition}}
\newtheorem{remark}{\textbf{Remark}}

\newtheorem{assumption}{\textbf{Assumption}}

\newacronym{rto}{RTO}{Real-Time Optimization}
\newacronym{gp}{GP}{Gaussian Processe}
\newacronym{kkt}{KKT}{Karush Kuhn Tucker}
\newacronym{dfo}{DFO}{Derivative-Free Optimization}
\newacronym{nlp}{NLP}{Nonlinear Program}
\newacronym{wor}{WOR}{Williams-Otto reactor}
\newacronym{rkhs}{RKHS}{reproducing kernel Hilbert space}

\newcommand{\R}{\mathbb{R}}

\newcommand{\Rnx}{\mathbb{R}^{n_x}}
\newcommand{\Rnxl}{\mathbb{R}^{{(n_x+1)} \times l}}

\newcommand{\mbb}[1]{\mathbb #1}

\newcommand{\eBox}{$\hfill\square$}

\newcommand{\rv}[1]{\mathsf{#1}}
\newcommand{\ProbMeasure}{\mathbb{P}}
\newcommand{\Outcomes}{\mathcal{D}}
\newcommand{\Ltwospace}[1]{L^2(\Outcomes, \ProbMeasure; \mathbb{R}^{#1})}

\usepackage{algpseudocode,algorithm,caption}

\DeclareCaptionFormat{algor}{%
	\hrulefill\par\offinterlineskip\vskip1pt%
	\textbf{#1#2}#3\offinterlineskip\hrulefill}
\DeclareCaptionStyle{algori}{singlelinecheck=off,format=algor,labelsep=space}
\newcounter{parentalgorithm}

\makeatletter

\title{\LARGE \bf
Convergence Certificate for Stochastic Derivative-Free Trust-Region Methods based on Gaussian Processes}

\author{Harsh A. Shukla$^{1}$, Tafarel de Avila Ferreira$^{2}$, Timm Faulwasser$^{3}$,  \\  Dominique Bonvin$^{1}$ and Colin N. Jones$^{1}$ 
	\thanks{The first and fifth author   acknowledge funding from the People Programme (Marie Curie Actions) of the EU Seventh Framework Programme (FP7/2007-2013) under REA grant agreement no 607957 (TEMPO).}
	\thanks{$^{1}$ Harsh A. Shukla, Dominique Bonvin, and Colin N. Jones are with Laboratoire d'Automatique, \'{E}cole Polytechnique F\'{e}d\'{e}rale de Lausanne, CH-1015 Lausanne, Switzerland.
		{\tt\small \{harsh.shukla, dominque.bonvin, colin.jones\}@epfl.ch}}%
	\thanks{$^{2}$ Tafarel de Avila Ferreira is with the Group of Energy Materials, \'{E}cole Polytechnique F\'{e}d\'{e}rale de Lausanne, CH-1951 Sion, Switzerland. \color{black}
		{\tt\small tafarel.deavilaferreira@epfl.ch}} %
	\thanks{$^{3}$ Timm Faulwasser is with the Institute for Energy Systems, Energy Efficiency and Energy Economics, TU Dortmund University, 44227 Dortmund, Germany. \color{black}
		{\tt\small timm.faulwasser@ieee.org}} %
}

\begin{document}

\maketitle
 \thispagestyle{empty}
\pagestyle{empty}

\begin{abstract}
In many machine learning applications, one wants to learn the unknown objective and constraint functions of an optimization problem from available data and then apply some technique to attain a local optimizer of the learned model. This work considers  \glspl{gp} as global surrogate models and utilizes them in conjunction with derivative-free trust-region methods. It is well known that derivative-free trust-region methods converge globally---provided the surrogate model is probabilistically fully linear. We prove that \glspl{gp} are indeed probabilistically fully linear, thus resulting in  fast (compared to linear or quadratic local surrogate models) and global convergence.  We draw upon the optimization of a chemical reactor  to demonstrate the efficiency of \gls{gp}-based trust-region methods.
\end{abstract}

\input{Tex_files/Introduction_new}

\input{Tex_files/Problem_formulation}
\input{Tex_files/Proposed_algorithm}
\input{Tex_files/Numerical_examples}

\section{CONCLUSIONS}
\label{sec:conclusions}
This paper has investigated convergence a certificate for stochastic derivative-free trust-region methods based on Gaussian Processes. To the best of our knowledge, this work is the first to show that \glspl{gp} are indeed probabilistic fully-linear models. This in turn allows inferring global convergence of trust-region methods in an  almost surely sense. We have demonstrated the efficacy of  \glspl{gp} as surrogate models, drawing upon repeated open-loop optimal control of a chemical batch reaction process.


 \bibliographystyle{IEEEtran}
\bibliography{references}

\end{document}

%% file: Tex_files/Introduction_new.tex

\section{INTRODUCTION}
\label{sec:introduction}

Increased computational power, ubiquitous availability of computational resources and improved algorithms have driven steady research interest in real-time optimization.  However, in essentially all real-world applications, accurate plant models are not available. Hence,  the issues surrounding uncertain models have been explored in different settings ranging from robust and stochastic optimization~\cite{sahinidis2004971} via real-time optimization~\cite{bonvin2017}, data-driven control~\cite{Ellis2014}  to machine learning~\cite{NING2019434}.

The set of \gls{dfo} trust-region methods comprises established tools to optimize unknown---or expensive to evaluate---objectives~\cite{Conn:2009:IDO:1508119}. The pivotal idea is the use of a local surrogate model, built at each iteration by evaluating the objective at a number of sample points within the trust region. 
Probabilistic derivative-free trust-region methods rely on randomized surrogate models~\cite{doi:10.1137/130915984, Larson:2016:SDO:2953787.2953863}. The key advantage of using a probabilistic method is its ability to capture uncertainties efficiently. This is indeed useful for noisy objectives and/or inaccurate models.
However, the key bottleneck of deterministic and probabilistic derivative-free trust-region methods alike is twofold: (i) ensuring the quality of the surrogate model, and (ii) guaranteeing a sufficiently large domain of validity. The former can be achieved via complicated procedures for sample-set maintenance~\cite{Conn:2009:IDO:1508119}, while the latter calls for global surrogate models. 

The convergence of trust-region methods relies on the accuracy of the surrogate model within the trust region. Intuitively speaking, the convergence mechanism increases the sampling of the unknown function and decreases the trust-region radius until the local surrogate model is sufficiently accurate in zeroth- and first-order compared to the unknown function. This accuracy, which is defined as ``full linearity", helps move in a decent direction. Global convergence of derivative-free trust-region methods for deterministic and stochastic version is described in~\cite{Conn:2009:IDO:1508119} and~\cite{doi:10.1137/130915984}, respectively. In this context, the main challenge is the construction of a surrogate model by performing as few plant evaluations as possible. Hence, if full-linearity  can be certified, a global surrogate model is usually preferred over local ones. 

At the same time, there is a recent and steadily growing interest in machine learning techniques in computer science as well as in systems and control. This spans $\{$supervised, reinforcement$\}$ learning  and data-driven function approximations by  deep neural networks~\cite{Lucia18a}  and \glspl{gp}~\cite{Hewing18a,Rasmussen:2005:GPM:1162254}. There exists  a body of literature on using machine learning for optimization, e.g.,~\cite{Rasmussen:2005:GPM:1162254, Shalev-Shwartz:2008}.  In some cases, it is possible to guarantee convergence to the global minimum of unknown functions. 

Since \glspl{gp} are excellent candidates to be used as global surrogate models, it is natural to combine them with derivative-free trust-region methods. The idea, which dates back to Conn's book~\cite{Conn:2009:IDO:1508119}, was analyzed empirically in~\cite{2017arXiv170304156A}. It is also used in \gls{rto}, where the aim is to solve a steady-state optimization problems or to optimize repeated batch operation~\cite{Ferreira2018RealTimeOO,Chanona19a}. 
However, to the best of our knowledge, it is yet to be shown whether \glspl{gp} can be certified to be fully linear, which is key for guaranteeing global convergence of derivative-free trust-region methods. 

The main result of this paper is to prove that \glspl{gp} can satisfy the fully-linearity property. We present the necessary procedure for this certification. Furthermore, using \gls{gp} as a global surrogate model leads to fewer trust-region iterations, implying fewer plant evaluations. This has a clear advantage over local surrogate models and other empirical local model correction methods as illustrated in the numerical section.

The remainder is structured as follows: We formulate the problem, discuss a probabilistic derivative-free trust-region method, introduce \glspl{gp}, and explain the need for improvement in Section~\ref{sec:problem_preliminaries}. A certification proof for \glspl{gp} to be fully-linear surrogate model is provided in Section~\ref{sec:proposed_algorithm}. We illustrate the effectiveness and advantage of the proposed work on a numerical case study in Section~\ref{sec:case_studies}, while Section~\ref{sec:conclusions} concludes the paper.

%% file: Tex_files/Problem_formulation.tex

\section{PRELIMINARIES}
\label{sec:problem_preliminaries}
\subsection{Problem statement}

We consider the \gls{nlp}
\begin{equation} \label{eq:NLP_General_unconstrained}
\underset{x \in \Rnx} {\min} ~ f({x}),
\end{equation}
with the unknown objective $f: \Rnx \to \R$ and the decision variables $x \in \Rnx$.
A solution to~\eqref{eq:NLP_General_unconstrained} can be computed using \gls{dfo} by  sampling the unknown function $ f $ and building a surrogate model. The samples are subject to additive noise and therefore their distribution can be written as: 
\begin{equation} \label{eq:f_sampling}
\rv{z} = f(x) + \rv{\nu} \quad \text{where } \nu \sim \mathcal{N}\left(0, \sigma^2 \right).
\end{equation}
Surrogate models usually depend (implicitly or explicitly) on a---yet to be specified---number of past data points, 
\begin{equation}\label{eq:Xdata}
\mbb{D}_k = \{{(x_{k-l-1}, z_{k-l-1})},\, \dots,\, (x_{k}, z_{k} )\},
\end{equation}
where $ z_{(k)} $ is a realization of the random variable $ \rv{z} $ at time instant $ k $.
Hence, by building the surrogate model $m: \Rnx\times\Rnxl  \to \mathbb{R}$, the solution to problem~\eqref{eq:NLP_General_unconstrained} becomes 
\begin{equation}\label{eq:NLP_General_approx_unconstrained}
x_{k+1} = \underset{x\in\Rnx} {\arg\min} \, m_k(x),
\end{equation}
where the shorthand $m_k(x):= m(x, \mbb{D}_k)$ is used.

%

\begin{remark}[Applicability of the considered setting] \label{rem:dynProbs}~\\
	At first glance, the problem setting outline above might look restrictive as it focuses on unconstrained optimization in real vector spaces. Although problem~\eqref{eq:NLP_General_approx_unconstrained} does not explicitly incorporate constraints, one may  convert  constrained optimization problems into unconstrained ones using penalty functions; see~\cite{NoceWrig06}.
	
	Moreover, whenever one aims at optimizing the performance of a repeated (batch) process or of a periodic process, 
$\dot{x} = f_p(x, v),\quad x(0) = x_0$,
	one will typically start with an optimal control problem, which---after applying direct discretization techniques in conjunction with ideas from sequential algorithms for numerical optimal control---can be cast in a mathematically equivalent form to~\eqref{eq:NLP_General_unconstrained} and~\eqref{eq:NLP_General_approx_unconstrained}, see e.g.~\cite{rawlings2017model}. 
 \eBox
\end{remark}
 
\subsection{Surrogate modeling with \glspl{gp}}
\label{sub:gp_intro}


Unlike parametric identification techniques, where data are discarded after constructing the model, \gls{gp}s are kernel-based methods that use all available data (or subset thereof) to learn a map between input and output data. 
We will briefly introduce \glspl{gp} and refer to~\cite{Rasmussen:2005:GPM:1162254, Murphy:2012:MLP:2380985} for further details.

Let $\rv{x}\in\Ltwospace{n_x}$ denote random variables and $x := \rv{x}(\omega)\in\Rnx$ their realizations, where $\Ltwospace{n_x}$ is the underlying Hilbert space of random variables with finite variance.


Considering $l$ available samples,  the input-output data generated by $f$ is $\mathbb{D}$ from~\eqref{eq:Xdata} (dropping the iteration index $k$). Let $\mathbb{X}$ be the projection of $\mathbb{D}$ onto $\Rnx$ and let $\bar{z}\in\R^{l}$ be the projection of $\mathbb{D}$ in direction of $z$. We use \gls{gp} regression to establish a relationship between $\mathbb{X}$ and $\bar{z}$ and obtain a corresponding conditional distribution of  $\rv{z}$ for a new query input point $x$, that is,
\begin{align}\label{eq:gp_m_v}
\rv{z}|x, \mathbb{X}, \bar{z} \sim \mathcal{N}\left(\mu_{\rv{z}}, \sigma_{\rv{z}}^2 \right),
\end{align}
where the mean and the variance of $z$ are
\begin{subequations}
\begin{align}
\mu_{\rv{z}} :=\mathbb{E}[\rv{z}] &= \bar{c}^\top \left( \bar{C}  + \sigma^2 I_p\right)^{-1}\bar{z}, \label{eq:gp_m}\\
\sigma_{\rv{z}}^2 := \mathbb{V}[\rv{z}] &= \kappa(x) -\bar{c}^\top\left( \bar{C}  + \sigma^2 I_p\right)^{-1}\bar{c}. \label{eq:gp_v}
\end{align}
\end{subequations}
Here, $\bar{C} \in \mathbb{R}^{l \times l}$ is a covariance matrix with elements $\bar{C}_{ij} = c(x_i, x_j)$,  $\bar{c} = \left[c(x,x_1), c(x,x_2), \ldots, c(x,x_l) \right]^\top  \in \mathbb{R}^{l \times 1}$, $\kappa(x) = c(x,x)$, where $c(\cdot,\cdot)$ is a covariance function denoted as  \textit{kernel}. For example, the squared exponential covariance function with automatic relevance determination is defined as
\begin{align}
c(x_i,x_j) = \sigma^2_f \;\text{exp}\textstyle \left( - \frac{\left( x_i-x_j\right)^\top \Lambda \left( x_i-x_j\right)}{2} \right),
\end{align}
where $\Lambda = \text{diag}( \lambda_1, \lambda_2, \ldots, \lambda_{n_x})$. The hyperparameters $\mathbf{\mathbf{\theta}} = \left[\sigma_f, \lambda_{1:n_x} \right] \in \mathbb{R}^{n_x+1}$ need to be learned/estimated from the data $\mathbb{D}$ during the training phase. Since the covariance matrix $\bar{C}$ and the covariance vector $\bar{c}$ depend on the hyperparameters $\theta$ and the input data $\mathbb{X}$, one can also write $\bar{C}$ as $\bar{C}(\mathbf{\theta}, \mathbb{X})$. To this end, consider $Q(\mathbf{\mathbf{\theta}}, \mathbb{X}):= \bar{C}(\mathbf{\mathbf{\theta}}, \mathbb{X})  + \sigma^2 I_l$ and the log-marginal likelihood 
\[
\mathcal{L}(\mathbf{\theta}, \mathbb{X}, \bar{z}) =  -\frac{1}{2} \bar{z}^\top Q(\mathbf{\theta}, \mathbb{X})^{-1}\bar{z}  -\frac{1}{2} \text{log} |Q(\mathbf{\theta}, \mathbb{X})| - \frac{n}{2} \text{log} 2\pi.
\]
Given $\mathbb{X}$ and $\bar{z}$, the parameters are learned by maximizing the log-marginal likelihood, 
\begin{equation} \label{eq:learning}
\mathbf{\theta^*}= \underset{\theta}{\mathrm{argmax }} \text{ } \mathcal{ L}(\theta, \mathbb{X},\bar{z}).
\end{equation}

The above nonconvex maximization problem can be solved using deterministic as well as stochastic methods. We refer to~\cite[Chap. 7]{Rasmussen:2005:GPM:1162254} and~\cite[Chap. 2]{kocijan2016modelling} for further details and for insights into the convergence properties.

In summary, the optimal hyperparameters $\mathbf{\theta^*}$ are found by training a \gls{gp}. Then, for the query point $x$, the \gls{gp} provides the output $\rv{z}$ with a normal distribution. The mean and covariance of the normal distribution is computed using the hyperparameters $\mathbf{\theta^*}$ and finding how close $x$ is compared to the data set. This way, the \gls{gp} computes the output distribution of $\rv{z}$ with more weight on the nearest inputs.

A first advantage of using \glspl{gp} compared to fixed-structure parameteric models is that GP models are able to capture complex nonlinear input-output relationships through the use of only a few parameters. This happens because the predicted output is influenced more by the nearby input-output pairs obtained from the training data set. A second advantage is that, generally, they offer an interesting trade-off between exploration and exploitation~\cite{Rasmussen:2005:GPM:1162254}. A third advantage, particularly in \gls{dfo} setting, is the fact \glspl{gp} constitute  global surrogate models. The main drawback of \gls{gp}s is the computational complexity growing as a cubic function of the number of data points $N$, that is, the complexity is $\mathcal{O}(N^3)$. However, this can be addressed using sparse \glspl{gp}~\cite{Rasmussen:2005:GPM:1162254}.



In what follows, at each iteration $k$, we use the \gls{gp} mean as a surrogate model, that is,
\begin{align}
	m_k(x) := \mu_{\rv{z}}\left(x, \mathbb{D}_k\right),
\end{align}
where the notation $\mu_{\rv{z}}\left(x, \mathbb{D}_k\right)$ highlights that, for fixed hyperparameters, $\mu_\rv{z}$ from \eqref{eq:gp_m} takes $x$ as argument---via $\bar c$---and depends on the data set $\mathbb{D}_k$--via $\bar c$ and $\bar C$. 

If the considered function samples obtained via~\eqref{eq:f_sampling} are indeed subject to additive noise, the data $\mathbb{D}_k$ will contain $l$ samples that correspond to realizations of random variables. Hence, the uncertainty surrounding the data $\mathbb{D}_k$ induces the probabilistic nature of the surrogate model and, consequently, the model available at iteration $k$ can be regarded as  $\rv{m}_k: \Rnx \to \Ltwospace{}$.  Conceptually, its realization can be denoted as $m_k :=  \rv{m}_k \left(\omega\right)$. This point of view leads naturally to probabilistic \gls{dfo} methods.

\subsection{Derivative-free probabilistic trust-region methods}
\label{subsec:derv_free_trust_reg}

A standard version of probabilistic derivative-free trust-region method is summarized in Algorithm~\ref{algo: trust_region_vanilla}, cf.~\cite{doi:10.1137/130915984, Larson:2016:SDO:2953787.2953863}. 
 The main idea is to approximate the unknown function via $m_k(x)$ within  a certain neighborhood of $x_{k}$ (a.k.a. the trust region). Whenever the surrogate model fails approximating the original problem, then the trust region is shrunk and the process repeated. 
 Next, we recall the main points of the convergence analysis given in~\cite{Larson:2016:SDO:2953787.2953863}.


\begin{assumption}[Differentiability of $f$~\cite{Larson:2016:SDO:2953787.2953863}]
	\label{assump: lipsc_cont}
The unknown function $ f $ has bounded level sets and the gradient $ \nabla f $ is Lipschitz continuous with constant $ L_g $.
	\eBox
\end{assumption}

\begin{assumption}[Noise with finite variance~\cite{Larson:2016:SDO:2953787.2953863}]
	\label{assump:noise_finite}
The additive noise $ \rv{\nu} $ observed while measuring $ f $  is drawn from a normal distribution with zero mean and finite variance. 
	\eBox
\end{assumption}

For the remainder, we define  $B(x;\Delta)$ as the ball of radius $ \Delta $ centered at $ x \in \mathbb{R}^n $. Furthermore, $ \mathcal{ C}^k $ denotes the set of functions on $  \mathbb{R}^n $ with $ k $ continuous derivatives and $ \mathcal{LC}^k $ denotes the set of functions in $\mathcal{ C}^k$  such that the $ k $th derivative is Lipschitz continuous. 

\begin{definition}[$ \kappa$ fully-linear model~\cite{Larson:2016:SDO:2953787.2953863}]
	\label{def:fully_linear}
Consider $f $ satisfying Assumption
\ref{assump: lipsc_cont}. 
	 Let $ \kappa = \left( \kappa_{ef}, \kappa_{eg}, \nu_{1}^{m} \right) $ be a given vector of constants and let $\Delta > 0 $ be given. A model $m\in \mathcal{LC}^1$ with Lipschitz constant $ \nu_{1}^{m} $ is a \textit{$\kappa$ fully-linear} model of $ f $ on $B(x;\Delta)$ if for all $s \in B(0;\Delta)$, 
\begin{subequations}	\label{eq:fully_linear_def}
\begin{align}
	|  f(x+s) -  m(x+s)|  &\leq \kappa_{ef} \Delta^2, \text{ and } \label{eq:fully_linear_def_zero} \\
	 \| \nabla f(x+s) - \nabla m(x+s) \| &\leq \kappa_{eg} \Delta. \label{eq:fully_linear_def_first}
\end{align}
\end{subequations}
 \eBox
\end{definition}

The above definition is  key in the convergence analysis for the case of probabilistic surrogate models. The main idea is to show that these models have good accuracy with sufficiently high probability~\cite{doi:10.1137/130915984}. Since derivative-free trust-region algorithms sample and collect data at each iteration, let $F_{k-1}^M$ denote the realization of events during the first {$k-1$} iterations of the algorithm. Now, we are ready to define a probabilistic $ \kappa $ fully-linear surrogate model. 
\begin{definition}[{$\kappa$ fully-linear model with probability  $\alpha $~\cite{Larson:2016:SDO:2953787.2953863}}]
	\label{def:probab_fully_linear}  
	A sequence of random models $ \{\rv{m}_k \} $ is  $ \kappa $ fully linear with probability $\alpha $ on $ \{B\left(x_k, \Delta_k\right)\} $ if the events
	\[ S_k = \{ m_k \text{ is a $\kappa $ fully-linear model of } f \text{ on }B\left(x_k, \Delta_k\right) \}\]
	satisfy the condition
	$\mathbb{P}\left(S_k|F_{k-1}^M\right) \geq \alpha$ 
	 for all $k$ sufficiently large. 
	\eBox
\end{definition}

Next, we introduce Algorithm \ref{algo: trust_region_vanilla}. The main idea is to build a surrogate model within the trust-region radius and use it to compute a minimizer. As long as the objective decreases sufficiently, accept the step and increase the trust-region radius, otherwise decrease the radius and reject the step. The challenge stems from the probabilistic nature of the surrogate model, in particular from the fact that the confidence in the model is probabilistic. This hinders
 increasing the trust-region radius significantly. Hence, it is important to have a relationship between the probability $ \alpha $ (confidence in the surrogate model) and $ \gamma_{inc}/\gamma_{dec} $ (increment/decrement of the radius). This relationship reads~\cite{Larson:2016:SDO:2953787.2953863}:

\begin{equation}
\label{eq:alp_relation_gamma}
\alpha \geq \left\{\frac{1}{2}, 1 - \frac{\frac{\gamma_{inc}-1}{\gamma_{inc}}}{4 \left[\frac{\gamma_{inc}-1}{2\gamma_{inc}} + \frac{1-\gamma_{dec}}{\gamma_{dec}}\right]}, 1 - \frac{1-\gamma_{dec}}{2\left( \gamma_{inc}^2-\gamma_{dec}\right)} \right\}.
\end{equation}

\begin{remark}
	A careful look at Step 1 of Algorithm~\ref{algo: trust_region_vanilla} reveals that we need to build a $ \kappa $ fully-linear model only for sufficiently large $ k $. This allows having a relatively inaccurate model at the beginning, thereby avoiding unnecessary sampling as long as there is sufficient improvement.
\end{remark}

\begin{algorithm}[t]
	
	\KwData{Initial model $m_{0}$, initial point $x_0$, and constants  $ 0 < \gamma_{dec} < 1 < \gamma_{inc}$,  $0 < \eta < \beta < 1 $, $ 0 < \Delta_0 $} and $ \alpha \in \left( 0 , 1 \right) $ satisfying~\eqref{eq:alp_relation_gamma}. Set k =0.
	\begin{enumerate}
		
		\item {\textbf{Model building: }} Build  $ m_k $, a $ \kappa $ fully-linear model with probability $\alpha_k $ on $  B(x_k; \Delta_{k}) $, for some $ \alpha_k \in (0,1) $ such that $ \alpha_k \geq \alpha $ for sufficiently large $ k $.

		\item \textbf{Step calculation:}
		
		\begin{equation}  
		s_{k} :=  \underset{s:\|s\| \leq \Delta_{k}} {\arg\min} \,
		m_{k}(x_{k} + s)
		\end{equation}
		
		\item \textbf{Compute model decrement: }
		
		\begin{enumerate}
			\item If $ m_{k}(x_k ) - m_{k}(x_k + s_k)  < \beta \text{ min } \left\{ \Delta_k , \Delta_k^2 \right\}$ then $ x_{k+1} = x_k $; $ \Delta_{k+1} = \gamma_{dec}\Delta_{k}$ and go to Step 6.
			\item Else go to Step 4).
		\end{enumerate}
		\item \textbf{Estimate improvement after plant evaluation: } Evaluate
		\begin{equation} 
		\rho_{k} = \frac{F_{k}^0 - F_{k}^{s_k}}{m_{k}(x_k ) - m_{k}(x_k + s_k)}.
		\end{equation}
		
		%

		\item {\textbf{Trust region and step update: }} 
		\begin{itemize}
			\item If $\rho_{k} \geq \eta$, then $x_{k+1} = x_{k} + s_{k}$ and $\Delta_{k+1} =  \gamma_{inc} \Delta_{k} $.		
			\item If $\rho_{k} < \eta $, then $x_{k+1} = x_{k}$ and $\Delta_{k+1} =  \gamma_{dec} \Delta_{k} $.	
			
			%
			
			

		\end{itemize}
		
		\item {\textbf{Setting index: }} k = k+1 and go to Step 1.
	\end{enumerate}

	\caption{ Derivative-Free Trust-Region Method~\cite{Larson:2016:SDO:2953787.2953863}}
	\label{algo: trust_region_vanilla}
\end{algorithm}

\vspace{10pt}

\begin{theorem}[Global convergence~\cite{Larson:2016:SDO:2953787.2953863}]\label{theorm:trust_vanilla_stocha_converge}
If Assumptions 1-2 are satisfied, and $ \alpha $ is chosen to satisfy~\eqref{eq:alp_relation_gamma}, then $ \left\{  \| \nabla f(x_{k}) \| \right\} $ converges in probability to zero. That is, for all $ \epsilon > 0 $,
	$\lim_{k\to\infty} \mathbb{P} \left[ \| \nabla f(x_{k}) \| > \epsilon  \right] = 0.$
\eBox
\end{theorem}
\vspace{10pt}



%

At this point a pivotal question arises: how to build a  $ \kappa $ fully-linear surrogate model with probability $\alpha $? Details of building and certifying a probabilistic local surrogate model at each iteration---mainly via linear and nonlinear interpolation/regression---are given in~\cite{doi:10.1137/130915984,Larson:2016:SDO:2953787.2953863}. Here, we aim at reducing the number of expensive plant evaluations by constructing a global instead of a local surrogate model. For that, we will use a \gls{gp} as the surrogate model. We will also show how to certify a \gls{gp} as a probabilistic fully-linear model. To the best of the authors' knowledge, this is still an open question, although \glspl{gp} have been used in a derivative-free trust-region framework~\cite{2017arXiv170304156A, Chanona19a}.  

%% file: Tex_files/Proposed_algorithm.tex
\section{CERTIFICATION PROOF}
\label{sec:proposed_algorithm}


\

We certify that \glspl{gp} are probabilistic fully-linear models. We remind the reader that we use the \gls{gp} mean  as the surrogate model, that is, $m(x) := \mu_{\rv{z}}\left(x, \mathbb{D}_k\right) $.

\vspace{10pt}
\begin{definition}[Reproducing kernel Hilbert Space~\cite{Rasmussen:2005:GPM:1162254}]
	\label{def:rkhs}
	Let $ \mathcal{H} $ be the Hilbert space
	of real functions $ f $ defined on the index set $ X $. Then, $ \mathcal{H} $ is called a \gls{rkhs} endowed with an inner product $ \langle\cdot,\cdot\rangle $ (and norm $ \| f \|_{\mathcal{H}} = \sqrt{\langle f,f \rangle_{\mathcal{H}}} $) if there exists a function $c : X \times X \rightarrow \mathbb{R}$ with the following properties:
	\begin{enumerate}
		\item for every $ x $, $c(x, x')$ as a function of $ x' $ belongs to $\mathcal{H}$, and
		\item $ k $ has the reproducing property $ \langle f(\cdot), c(\cdot, x)\rangle_{\mathcal{H}} = f(x) $.
	\end{enumerate}
	\eBox
\end{definition}

The aim is to show that~\eqref{eq:fully_linear_def_zero} and~\eqref{eq:fully_linear_def_first} hold with probability at least $\alpha $ when the \gls{gp} mean is used as a surrogate model. For this, the following two properties are assumed.
 \vspace{10pt}
\begin{assumption}[Bounded \gls{rkhs} norm~\cite{6138914}]
	\label{assump:f_rkhs_norm}~\\
	The unknown function $f(x)$ has a known bounded \gls{rkhs} norm $\zeta$ under a known kernel $c$, that is, $\| f(x)\|_{c} \leq \zeta < \infty$. \eBox
\end{assumption}

\color{black}
\vspace{10pt}
\begin{assumption}[Lipschitzness of the mismatch function]
	\label{assump:lip_mismatch_grad}
	The mismatch function $h(x) := f(x)- m(x) $ has Lipschitz continuous gradient with constant $\gamma_{lh}$. Furthermore, the sequence $x_k$ generated by applying Algorithm~\ref{algo: trust_region_vanilla} satisfies  $\| \nabla^2 h(x_k) \| \leq \kappa_{bhh} < \infty$, that is, the mismatch function has a bounded Hessian. \eBox
\end{assumption}
\vspace{10pt}

Assumption~\ref{assump:lip_mismatch_grad} is not very strong and is a consequence of Assumption~\ref{assump: lipsc_cont}: the unknown function $f(\cdot)$ has Lipschitz continuous gradient with bounded Hessian. Note that most of the practically used kernels (e.g. Matern, squared exponential) have Lipschitz continuous gradients~\cite{Rasmussen:2005:GPM:1162254}. Before deriving the main result, we first state that the distance between an unknown function and the mean is bounded by the \gls{gp} variance  with some probability $1 - \delta$.


%
%


\begin{lemma}[Bound on mismatch function~\cite{6138914}]
	\label{lemma:GP_mean_bound}	~\\
	Let Assumption~\ref{assump:f_rkhs_norm} holds and let $\delta \in (0,1)$. It follows that
	$\mathbb{P} \left\{ |m(x) - f(x)| \leq \sqrt{\beta(N,\delta)}\, \sigma_z(x,N) \right\}   \geq 1 - \delta$.\eBox
	
\end{lemma}

Here, $\sqrt{\beta(N,\delta)}$ depends on the number of samples, the probability $\delta$ and the \gls{rkhs} norm $\| f\|_{c}$, see~\cite{6138914} for details. If the unknown function $f$ is sampled from a \gls{gp}, one can compute $\beta$ in closed form~\cite{2019arXiv190601376L}. We note that for highlighting the dependence of $ \beta $ and $ \sigma_{z} $ on the number of  samples and the probability $ \delta $, we simply write them as $ \beta(N,\delta) $ and $ \sigma_z(x,N) $.

%

\vspace{10pt}
\begin{theorem}[\gls{gp} is $ \kappa$ fully linear  with probability $\alpha$]
	\label{theorm:gp_probab_fully_linear}
	Let Assumptions~\ref{assump:f_rkhs_norm} and~\ref{assump:lip_mismatch_grad} hold. If $ 0 < \Delta < \frac{6}{\gamma_{lh}} \left(\kappa_{eg} - 2 \kappa_{ef} - \kappa_{bhm}\right) $, then there exists a positive integer $ N < \infty $ such that, after $ N $ sampling steps, a \gls{gp} can  be certified  $ \kappa $ fully linear with probability $\alpha$.	\eBox
\end{theorem}

\begin{proof}	
	Following Definition~\ref{def:probab_fully_linear}, the goal is to prove that equations~\eqref{eq:fully_linear_def_zero} and~\eqref{eq:fully_linear_def_first}  hold with probability at least $\alpha$. 
	
	Let us start with equation~\eqref{eq:fully_linear_def_zero} and  consider any point within the trust region, that is, $ x \in B(x_k, \Delta_k) $. Increased sampling will validate the probability bound in Lemma~\ref{lemma:GP_mean_bound}. Upon performing $ N $ plant evaluations and applying Algorithm 1 in~\cite{6138914} with $ \alpha \leq 1-\delta $, the following holds with  probability $\alpha$ for a given $  \kappa_{ef} $ and $ \Delta $: 

	\begin{equation}
	\begin{aligned}
	\label{eq: hx_variance_zero}
	|h(x)| = 	|m(x) - f(x)|  \leq \sqrt{\beta(N,\delta)} \,\sigma_z(x,N) \leq  \kappa_{ef}  \Delta^2,
	\end{aligned}
	\end{equation} 
	which certifies equation~\eqref{eq:fully_linear_def_zero} with probability $ \alpha $.
	
	Next we turn to equation~\eqref{eq:fully_linear_def_first} and take any $ x, x_s \in B(x_k, \Delta_k) $ such that $ x_s = x + s $. Taylor's expansions give:
	\begin{align}
	h(x+s) =& h(x) + s^{\top} \nabla h(x) + s^{\top} \nabla^2 h(x) s + \mathcal{O}(s^3) \notag \\
	|s^{\top} \nabla h(x)| =& |h(x+s) - h(x) - s^{\top} \nabla^2 h(x) s - \mathcal{O}(s^3)| \notag \\
	\leq& |h(x+s)| + |h(x)| + |s^{\top} \nabla^2 h(x) s| + |\mathcal{O}(s^3)| \notag  \\
	\leq& |h(x+s)| + |h(x)| + |s^{\top} \nabla^2 h(x) s| + \frac{\gamma_{lh}}{6} \|s\|^3 \notag ~,
	\end{align}
		where the first inequality comes from norm properties and the second using Lemma 4.1.14 in~\cite{highertaylor}. Substituting
	$ s := \frac{ \nabla h(x)  \Delta}{\| \nabla h(x) \|} $ by following Lemma 4.7 in~\cite{doi:10.1137/18M1173277} gives,
	\begin{align*}
	\Delta \| \nabla h(x) \| \leq& |h(x+s)| + |h(x)| +    \Delta^2 \| \nabla^2 h(x) \| + \frac{\gamma_{lh}}{6} \Delta^3\\
	\Delta \| \nabla h(x) \| \leq& |h(x+s)| + |h(x)| +    \kappa_{bhh}\Delta^2   + \frac{\gamma_{lh}}{6} \Delta^3. \notag 
	\end{align*}	
	Here, the last inequality arises because of Assumption~\ref{assump:lip_mismatch_grad}. As shown in the first part of this proof, one can guarantee that   $ |h(x+s)|, | h(x)| \leq  \kappa_{ef} \Delta^2$ with at least  probability $ 1-\delta $. Hence,  the following holds with with  probability at least $ (1-\delta)^2 $:	
	\begin{align}
	\Delta \| \nabla h(x) \| \leq&   2\kappa_{ef} \Delta^2 +   \Delta^2  \kappa_{bhh} + \frac{\gamma_{lh}}{6} \Delta^3 \notag \\
	\| \nabla h(x) \| \leq&    2\kappa_{ef} \Delta +   \Delta  \kappa_{bhh} + \frac{\gamma_{lh}}{6} \Delta^2. \notag
	\end{align}
	Choosing $ \delta $ such that $ \alpha \leq (1-\delta)^2 $ and combining the above with Definition~\ref{def:probab_fully_linear} and~\eqref{eq:fully_linear_def_first},
	it remains to show that, for a given $ \kappa_{eg}$, the following criterion can be satisfied: 
\[	2 \kappa_{ef} \Delta +   \Delta  \kappa_{bhh} + \frac{\gamma_{lh}}{6} \Delta^2 \leq  \kappa_{eg} \Delta.\]
		 Since $ 0 < \Delta < \frac{6}{\gamma_{lh}} \left(\kappa_{eg} - 2 \kappa_{ef} - \kappa_{bhm}\right) $, the above inequality is satisfied. Hence, equation~\eqref{eq:fully_linear_def_first} holds with 
		 probability  at least $ \alpha$, which concludes the proof. 
\end{proof}
\vspace{10pt}


\begin{remark} [Computing $ \beta $ and finding maximum $ \sigma_z(x) $]
	 $\sqrt{\beta(N,\delta)}$ is not a function of $x$. However, we need to determine the maximum of $ \sigma_z(x,N) $ over $x$ within the trust region. This problem has been tackled rigorously in the machine learning community, see~\cite{6138914} for details. However, for our application, one need not explicitly  compute these quantities. Another way to look at it is that one can always choose arbitrarily large $ \kappa_{ef} $ and $ \kappa_{eg} $ such that~\eqref{eq:fully_linear_def_zero} and~\eqref{eq:fully_linear_def_first} are satisfied with probability at least $ \alpha $. 
	\eBox
\end{remark}

\begin{remark}[Condition on $ \Delta $ in Theorem~\ref{theorm:gp_probab_fully_linear}]~\\
The condition on the trust-region radius $ \Delta $ in  Theorem~\ref{theorm:gp_probab_fully_linear} does not limit/restrict the algorithm significantly. The reason is that one can choose arbitrarily large values of $ \kappa $. Moreover, the trust-region radius almost surely goes to zero [Lemma 4~\cite{Larson:2016:SDO:2953787.2953863}]. Hence, for any positive $ \kappa $, the condition on the trust region is  almost surely satisfied. 	\eBox
\end{remark}

Using \gls{gp} in  the framework of Algorithm~\ref{algo: trust_region_vanilla} yields almost surely convergence. Moreover, it has two main advantages:
(i) since \glspl{gp} approximate unknown functions globally, one does not need to sample after each trust-region iterations as opposed to standard trust-region approaches, where $n$ and $\frac{(n+1)^2}{2}$ data points are required for linear interpolation and nonlinear polynomial-based regression, respectively. This saves a significant amount of plant evaluations; 
(ii)  from an implementation point of view, there is no need to build a model at each trust-region iteration. This is due to the fact that the \gls{gp} mean  converges to the exact function in the limit, as per Lemma~\ref{lemma:GP_mean_bound} and for $ \alpha_k > \alpha $  for sufficiently large $ k $. In fact, to implement the algorithm after a failed iteration, one simply needs to sample (not necessarily in the trust-region radius) a few points. Hence, the computation of $ \beta(N,\delta) $ and of the maximal variance $ \sigma_z(x,N) $ mentioned in the proof of Theorem~\ref{theorm:gp_probab_fully_linear} can be avoided. This obviously reduces the computational burden.

%% file: Tex_files/Numerical_examples.tex
\section{NUMERCIAL CASE STUDY}
\label{sec:case_studies}

We apply the proposed method to the acetoacetlytation of pyrrole with diketene \cite{CHACHUAT20091557}. It is a batch-to-batch optimization of a semi-batch reactor process with 4 reactions: $ \text{A} + \text{B}  \overset{k_1}{\longrightarrow}   \text{C} $, $  2\text{B}  \overset{k_2}{\longrightarrow}   D $, $  \text{B}  \overset{k_3}{\longrightarrow}    \text{E} $, and $ \text{B} + \text{C}  \overset{k_4}{\longrightarrow}   \text{F} $.
The involved species are \text{A}: pyrrole; \text{B}: diketene; \text{C}: 2-acetoacetyl pyyrole; \text{D}: dehyroacetic acid; \text{E}: oligomers; \text{F}: undesired by-product. The material balance equations for the plant read~\cite{CHACHUAT20091557}: 
 \begin{equation}
 \label{eq:diek_mat_balance_eq}
\begin{aligned}
\dot{c}_{A}  &= -k_1 c_{A}c_{B} - \frac{F}{V} c_{A}, \\
\dot{c}_{B} & = -k_1 c_{A}c_{B} - 2k_{2}c_{B}^2 -k_3c_{B} - k_4 c_{B}c_{c} + \frac{F}{V} \left(c_{B}^{in} - c_{B}\right),\\ 
\dot{c}_{C}  &= ~~ k_1 c_{A}c_{B} -k_4 c_{B}c_{C} - \frac{F}{V} c_{C},\\
\dot{c}_{D}  &= ~~ k_{2}c_{B}^2- \frac{F}{V} c_{D}, \quad \dot{V} = F.
\end{aligned}
\end{equation}
%
%
%
It is assumed that the last two reactions are unknown (structural mismatch), thereby leading to the following plant model: $ \text{A} + \text{B}  \overset{k_1}{\longrightarrow}   \text{C} $, $ 2\text{B}  \overset{k_2}{\longrightarrow}   D $.

 We are interested in finding the feed profile of species \text{B} such that it maximizes the amount of \text{C} at final time, while maintaining the concentration of \text{B} and \text{D} below specified threshold values at terminal time. The resulting problem is:
 \begin{align}
 \underset{F(t)}  {\mbox{max}} \;&  J := c_C(t_f)V(t_f) \nonumber \\
\mbox{subject to: } & \text{model equations~\eqref{eq:diek_mat_balance_eq}}, \label{eq:diek_opti} \\
 c_B(t_f) \leq c_{B}^{max} \enspace,\quad  &c_D(t_f) \leq c_{D}^{max} \enspace, \quad 0 \leq F(t) \leq F^{max} \enspace. \nonumber
\end{align}

The optimal input profile is assumed to have three parts: a first arc with $ F = F^{max} $, a second arc where the feeding is between $ 0 $ and $ F^{max} $, and a third arc with zero feeding, see~\cite{CHACHUAT20091557} for details.  Accordingly, we can define three decision variables, namely, $ \pi := (t_m, t_s, F) $, where  $ t_m $ represents the switching time between the first and second arcs,  $ t_s $ the switching time between the second and third arcs, and $ F $ the assumed constant feeding rate during the second arc. We consider two different scenarios as listed in Table~\ref{tab:two_scenarios_CSS}. We compare the three-arc solution with 100 piecewise-constant control parametrization in Figure~\ref{fig:diek_input_comparisons}. Since the two control parameterizations offer similar performance, we use the three-arc parametrization in this work. We reformulate Problem~\eqref{eq:diek_opti} in unconstrained optimization by incorporating the constraints as a penalty term in the cost function. The \gls{gp} learns the mismatch between the plant and model costs. The plant measurements are corrupted with $ 5\% $ zero-mean Gaussian additive noise.

The \gls{gp} is trained using 20 random points around the model optimum for Scenario I. The \gls{gp} is assumed to be unaware of the change from Scenario I to Scenario II at batch/iteration 8. After each step of the trust-region algorithm, a new data point is incorporated in the \gls{gp} if it is sufficiently far from the previous data to avoid overfitting. Parameters of Algorithm~\ref{algo: trust_region_vanilla} are: $ \eta = 0.5 $, $ \gamma_{dec} = 0.9 $, $ \gamma_{inc} =3 $, and $ \Delta_0 = 3.5$.

Plant efficiency deteriorates significantly for Scenario II when the model-based optimized input is applied as shown in Figure~\ref{fig:diek_effi}. Furthermore, although the \gls{gp} model is unaware of the change in scenarios, it quickly learns the new plant operating condition (Figure~\ref{fig:diek_effi}), and significantly outperforms the model-based approach (Figure~\ref{fig:diek_input}).

\begin{table}[!htbp]
	\begin{center}
		\begin{tabular}{@{}llll@{}} \toprule
			Scenario & $ k_3 $ & $ k_4 $ & Batch  \\ \midrule
			Scenario I  & 0.01 & 0.009 & 1-7 \\ 
			Scenario II  & 0.28 & 0.001 & 8-22 \\   \bottomrule
		\end{tabular}
		\caption{Uncertain reaction constants and batch numbers for the two scenarios.}
		\label{tab:two_scenarios_CSS}
	\end{center}
\end{table}

\begin{figure}[!htb]
	\centering
	\	
	\includegraphics[scale=0.4]{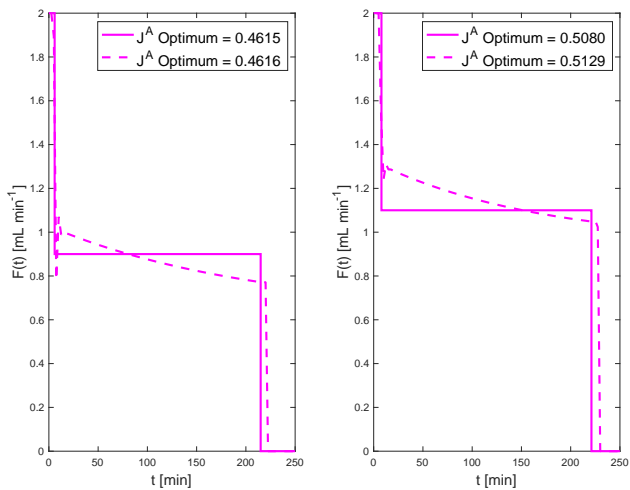}
	\caption{Input profiles for two different control parameterizations} 
	\label{fig:diek_input_comparisons}
	
\end{figure}

\begin{figure}[!htb]
	\centering
	\	
		\includegraphics[scale=0.27]{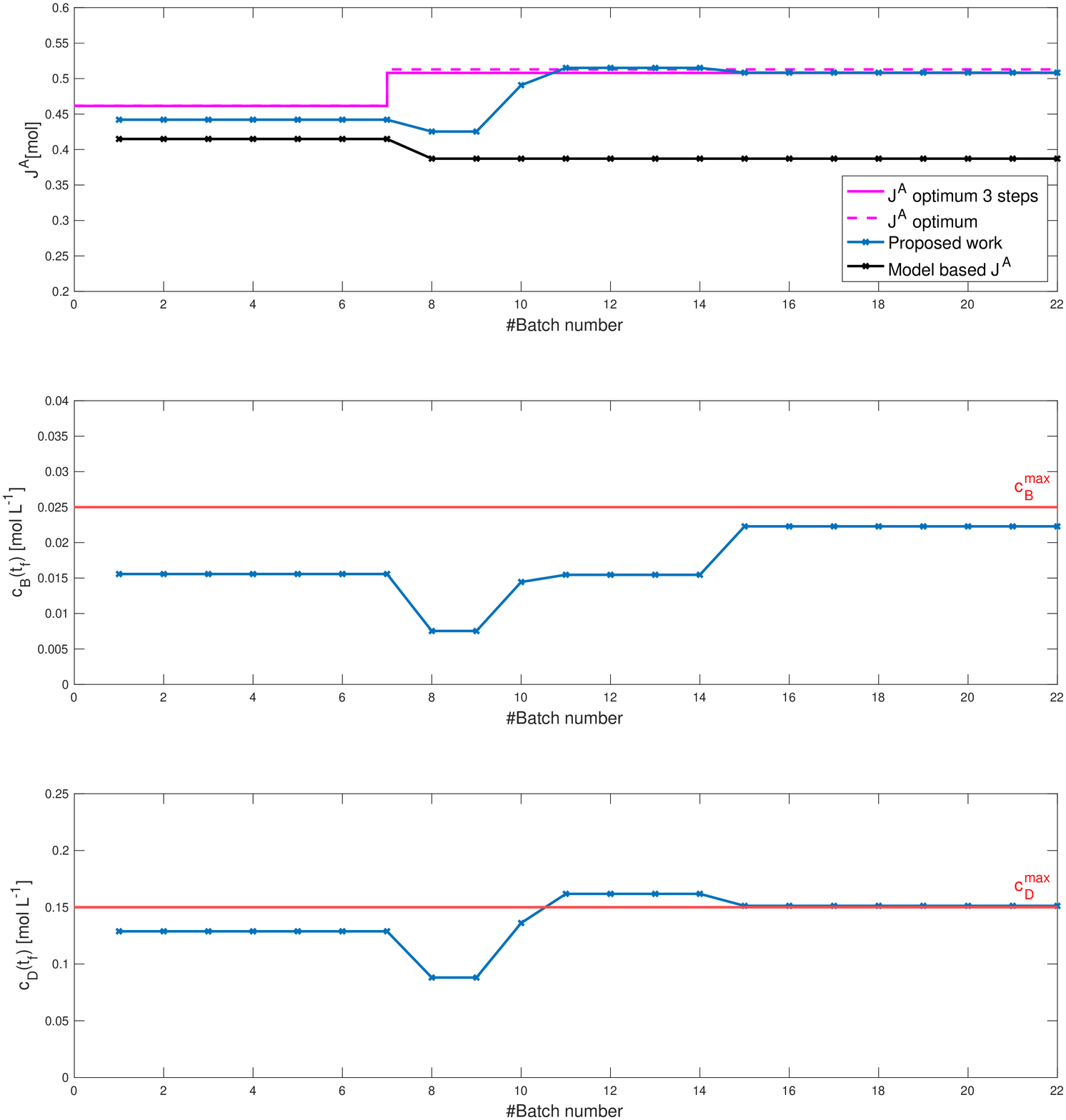}
		\caption{Learning and optimization using \gls{gp} as a surrogate model. The upper plot illustrates the cost profile, while the middle and lower plots show the constraints profiles. \gls{gp} reaches near plant optimality within three  iterations.}
		\label{fig:diek_effi}
	
\end{figure}

\begin{figure}[!htb]
	\centering
	\	
	\includegraphics[scale=0.4]{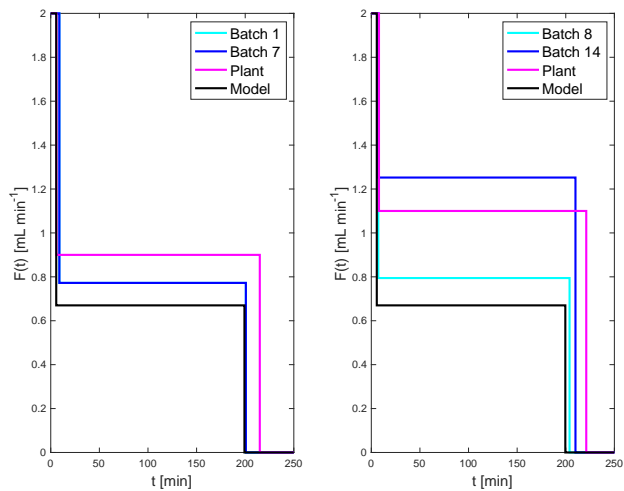}
	\caption{Input profiles for the two scenarios}
	\label{fig:diek_input}
	
\end{figure}